\documentclass[letterpaper,10pt,english]{article}

\usepackage[T1]{fontenc}
\usepackage{babel}
\usepackage{hyperref}

\usepackage{authblk}
\usepackage{amsmath}
\usepackage{amssymb}
\usepackage{amsthm}
\usepackage{tikz}
\usepackage{graphicx}
\usepackage{float}
\usepackage[ruled]{algorithm}
\usepackage[noend]{algpseudocode}
\usepackage{tabularx}
\usepackage[margin=1in]{geometry}
\usepackage{arydshln}
\usepackage{multirow}
\usepackage{wrapfig}
\usepackage{url}

\usetikzlibrary{positioning,calc,shapes,arrows}
\usetikzlibrary{backgrounds}
\usetikzlibrary{matrix,shadows,arrows} 
\usetikzlibrary{decorations.pathreplacing}

\def\etal.{et\penalty50\ al.}
\theoremstyle{plain}
\newtheorem{theorem}{Theorem}[section]
\newtheorem{lemma}[theorem]{Lemma}

\newtheorem{corollary}[theorem]{Corollary}

\theoremstyle{definition}
\newtheorem{definition}{Definition}[section]

\theoremstyle{remark}
\newtheorem{question}{Question}[section]

\newtheorem{observation}[question]{Observation}

\theoremstyle{plain}
\newtheorem*{theorem*}{Theorem}

\author[1]{Jurek Czyzowicz}
\author[2]{Ryan Killick} 
\author[2]{Evangelos Kranakis}
\author[3]{Danny Krizanc}
\author[4]{Lata Narayanan}
\author[4]{Jaroslav Opatrny}
\author[4]{Denis Pankratov}
\author[5]{Sunil Shende}

\affil[1]{D\'{e}partemant d'informatique, Universit\'{e} du Qu\'{e}bec en Outaouais,  Gatineau, Canada\\ \href{}{jurek.czyzowicz@uqo.ca}}
\affil[2]{School of Comp. Sci., Carleton University, Ottawa, Canada\\ \href{}{ryankillick@cmail.carleton.ca,kranakis@scs.carleton.ca}}
\affil[3]{Department of Mathematics \& Comp. Sci., Wesleyan University, Middletown CT, USA \\ \href{}{dkrizanc@wesleyan.edu}}
\affil[4]{Department of Comp. Sci. and Software Eng., Concordia University, Montreal, Canada\\ \href{}{lata,opatrny,denisp@cse.concordia.ca}}
\affil[5]{Department of Comp. Sci., Rutgers University, USA\\ \href{}{shende@camden.rutgers.edu}}

\title{Group Evacuation on a Line by Agents with Different Communication Abilities\thanks{This research is supported by NSERC, Canada.}}

\date{}

\newcommand{\Alg}{\mathcal{A}}

\newcommand{\CR}{\textsc{CR}}
\newcommand{\eps}{\epsilon}
\newcommand{\floor}[1]{\left\lfloor{#1}\right\rfloor}

    \renewcommand{\algorithmicrequire}{\textbf{Input:}}

    \newcommand{\Begin}{\renewcommand{\algorithmicrequire}{\textbf{Begin:}} \Require}
    \newcommand{\End}{\renewcommand{\algorithmicrequire}{\textbf{:End}} \Require}
    
    \algblockdefx{MRepeat}{EndRepeat}{\textbf{repeat}}{}
    \algnotext{EndRepeat}

\begin{document}

\maketitle

\begin{abstract}
We consider evacuation  of a group of $n \geq 2$ autonomous mobile agents (or robots)  from an unknown exit on an infinite  line. The agents are initially placed at the origin of the line and  can move with any speed up to the maximum speed $1$ in any direction they wish and they all can communicate when they are co-located. However, the agents have different wireless communication abilities: while some are fully wireless and can send and receive messages at any distance, a subset of the agents are  {\em senders}, they can  only transmit messages wirelessly, and the rest are  {\em receivers}, they can only receive messages wirelessly. The agents start at the same time and their communication abilities are known to each other from the start. Starting at the origin of the line, the goal of the agents is to collectively find a target/exit at an unknown location on the line while minimizing the {\em evacuation time}, defined as the time when the last agent reaches the target. 

We investigate the impact of such a mixed communication model on evacuation time on an infinite line for a group of cooperating agents. In particular, we provide evacuation algorithms and analyze the resulting competitive ratio ($CR$) of the evacuation time for such a group of agents. If the group has two agents of two different types, we give an  {\em optimal} evacuation algorithm with competitive ratio $CR=3+2 \sqrt{2}$. If there is a single sender or fully wireless agent,  and multiple receivers we prove that $CR \in [2+\sqrt{5},5]$, and if there are multiple senders and a single receiver or fully wireless agent, we show that $CR \in [3,5.681319]$.  Any group consisting of only senders or only receivers requires competitive ratio 9, and any other combination of agents has competitive ratio 3. 

\end{abstract}

\section{Introduction}

Search by a group of cooperating autonomous mobile robots for a target in a given domain is a fundamental topic in the theoretical computer science. In the  search problem one is interested in finding a target at an unknown location as soon as possible. In the related  {\em evacuation problem} one is interested in optimizing the time it takes the  last robot in the group to find the target, often called the {\em exit}. There has been a lot of interest in trying to understand the impact of communication between agents on the search and evacuation time in the distributed computing area. The design of optimal robot trajectories leading to tight bounds depends not only on the fault-tolerant characteristics of the agents but also on the communication model employed (see \cite{isaacCzyzowiczGKKNOS16,PODC16}).  In previous works, agents are assumed to either have full wireless communication abilities, i.e., they can both transmit and receive messages across any distance \cite{chrobak2015group}, or limited distance \cite{Bagheri2019}, or they have no wireless communication abilities, and can only communicate when they are {\em face-to-face (F2F)}, i.e., co-located. In terms of communication abilities, the agents are identical. 

The present work considers 
evacuation on an infinite line by a group $G$ of cooperating robots (initially located at the origin) whose wireless communication abilities are different, which compel them to employ a {\em mixed} communication model. 
At a {\em rudimentary} level they can always communicate reliably using F2F. 
However, some agents in $G$ are  {\em senders} that can transmit messages wirelessly at any distance but  only receive F2F, yet others  are  {\em receivers} in that they can receive messages wirelessly from any distance but can transmit only F2F, and the remaining are fully wireless,  and can both send and receive messages wirelessly. This situation might occur because it is cheaper to build agents with limited wireless capabilities, or because the sender or receiver module failed in receiver or sender robots, respectively.  Further, we assume the capabilities of the robots are known to each other in advance and remain the same for the duration of an 
evacuation algorithm. Robots  can move at any  speed up to maximum $1$. We give upper and lower bounds on the competitive ratio of evacuation algorithms, depending on the number of senders and receivers among the agents.




If there are at least two fully wireless agents in the group, then 
the optimal competitive ratio is 3, see \cite{chrobak2015group}. By pairing up a sender and a receiver we can simulate a fully wireless agent.  Consequently,   if there is one fully wireless agent, one sender agent, and one receiver agent, the competitive ratio is 3.      Consider now the case when there is one fully wireless agent, and one or more senders. Since the sender agents cannot receive wireless transmissions, the sending capabilities of the fully wireless agent are useless, and it is equivalent to a receiver agent. Similarly if there is one fully wireless agent, and one or more receivers,  the receiving module in the fully wireless agent is useless, and it is equivalent to a sender agent.

     Thus we no longer consider fully wireless agents, and only consider sender and receiver  agents. 
     When all of the agents are senders, or all of them are receivers, the only possible mode of communication between agents is F2F;  in this case, it has previously been demonstrated that the optimal competitive ratio of evacuation for $n$ F2F agents is $9$.   If  there are at least two sender agents and two receiver agents,  by pairing up sender and receiver agents, we obtain a competitive ratio of 3.       
    It follows that the 
  only interesting cases to consider are when there is exactly one sender and one receiver; one sender and several receivers; one receiver and  several senders. These are the cases investigated in detail in this paper. 


\subsection{Model and preliminaries}
        We consider the problem of evacuation by $n \geq 2$ mobile agents beginning at the origin of the infinite line. All agents are assumed to have maximum speed 1 and can move in either the positive direction (referred to as moving to the right) or the negative direction (referred to as moving to the left). The agents may change their speed and the direction of motion instantaneously and arbitrarily often. Moreover, the robots can choose any speed as long as it does not exceed the maximum speed 1. 

All agents have the ability to communicate F2F, 
however the wireless communication abilities of the agents are limited and are not all the same. Indeed the group of $n$ agents consists of a subset of agents that can only  send wireless messages, called {\em senders}, and a subset that can only receive wireless messages, called {\em receivers}. 
We represent by $n_s \geq 0$ and $n_r = n-n_s$ the number of senders and receivers respectively.

        The cost of an algorithm  for the evacuation problem on a given instance of the problem is the time the {\em last} agent  reaches  the target, called  the {\em evacuation time}.  We denote by $E(\Alg,x)$ the evacuation time of algorithm $\Alg$ when the target is at location $x$. Note that an offline algorithm in which agents know the position of the target can reach it in time $|x|$. The goal is to minimize the {\em competitive ratio}, denoted by $\CR$, defined as  the supremum, over all possible target locations, of the normalized cost $E(\Alg,x)/|x|$, i.e., 
                $\CR(\Alg) := \sup_{|x| > 1}\frac{E(\Alg,x)}{|x|}.$

        An evacuation algorithm can be primarily viewed as a set of trajectories, one for each agent. The trajectory of an agent specifies where the agent should be located at any given time.  More specifically, the trajectory of an agent is a continuous mapping from the non-negative reals (i.e. time) to the reals (i.e., position on the line). In general, we will represent the trajectory of an agent using the notation $X = X(t)$ with the interpretation that the agent with trajectory $X$ will be located at position $X(t)$ at time $t$.
        Due to our assumption that the agents have maximum unit speed, an agent trajectory $X$ must satisfy
                $|X(t')-X(t)| \leq t'-t,\quad \forall t' \geq t \geq 0.$ 
        Agents are assumed to begin their search at the origin and so we must also have
                $X(0) = 0.$ 
        Taken together, these equations imply that 
                $|X(t)| \leq t,\ \forall t \geq 0.$ 
        
         We assume that the agents are labelled so that we may assign a specific trajectory to a specific agent. Each agent is assumed to know the trajectory of all  other agents. 
        All agents follow their assigned trajectories until they either find the target or are otherwise notified of the target's location. What an agent does in the event that it finds the target depends on the communication ability of the finder and of the other agents. For example, if the finder is a sender and all other agents are receivers, then the sender can immediately notify the other agents who can then proceed to move at full speed to the target. On the other hand, if the finder is a receiver, then it must move to notify another agent(s) of the target's location. In any case, the cost of the algorithm will depend both on the trajectories assigned to the agents to search for the target, as well as the subsequent phase of informing all agents, and  the agents travelling to the target.  


%
\subsection{Related work}        

Search by a single agent on the infinite line was initiated independently by Beck \cite{beck1964linear,beck1965more,beck1970yet} and Bellman \cite{bellman1963optimal}. These seminal papers proved the competitive ratio $9$ for search on an infinite line and also gave the impetus for additional studies, including those by Gal \cite{gal1972general} which proposes a minimax solution for a game in which player I chooses a real number and player II seeks it by choosing a trajectory represented by a positive function, Friestedt~\cite{fristedt1977hide}, Friestedt and Heath~\cite{fristedt_heath_1974}, and Baeza-Yates et.al.\cite{baezayates1993searching,baezayates1995parallel} where search by agents in domains other than the line were proposed, e.g. the ``Lost Cow'' problem in the plane or at the origin of $w$ concurrent rays. Additional information on search games and rendezvous can also be found in the book~\cite{alpern2006theory}.

Group search on an infinite line has been researched in several papers and under various models. Evacuation by multiple cooperating robots was proposed in \cite{chrobak2015group}, for the case where the robots can communicate only F2F. More recently, search on the line was considered for two robots which have distinct speeds in \cite{bampas2019linear} and in \cite{demaine2006online} when turning costs are taken into account. In addition, in two papers \cite{CGKKKLNOS19wireless,czyzowicz2021time} the authors are concerned with minimizing the energy consumed during the search.

There are several types of robot communication models in the literature. The most restricted type of robot communication is F2F in which robots may exchange messages only when they are co-located. At the other extreme is wireless in which robots may communicate regardless of how far apart they are~\cite{georMAC}. A model where the wireless communication range is limited has been explored  for the equilateral triangle domain in \cite{Bagheri2019}.

Within these communication models researchers have considered search with crash~\cite{PODC16} and Byzantine~\cite{isaacCzyzowiczGKKNOS16} faults. The former are innocent faults caused by robot sensor malfunctioning causing the robots an inability to communicate and/or perform their tasks. The latter, however, are malicious faults (intentionally or otherwise) in that the robots may lie and communicate maliciously the wrong information. 
Lower bounds for search in the crash fault model are proved in~\cite{kupavskii2018lower} and for Byzantine faults in~\cite{Sun2020}. The competitive ratio for search and evacuation in the near majority case (of $n=2f+1$ robots with $f$ faulty) is a notoriously hard problem and additional results can be found in the recent paper~ \cite{ryanSIAM}. Additional information and results can also be found in the recent PhD thesis~\cite{phdthesisryan}. 

In our paper we investigate evacuation time by agents with different wireless communication abilities; as stated earlier, some of them can only transmit wirelessly but not receive, and the others can only receive messages sent wirelessly but not transmit. To our knowledge, in all previous works, the communication abilities of agents were identical;  group search or evacuation by agents of  different communication abilities has not been studied before. 

\subsection{Results}
As mentioned above,  we need to consider three cases: when there is exactly one sender and one receiver; one sender and several receivers; one receiver and  several senders.
        In Section~\ref{sec:upper_bounds} we give evacuation algorithms and
analyze upper bounds on the competitive ratios for each of these three  cases. When we have one receiver and one sender agent, i.e., $n_s=n_r=1$, we give an evacuation algorithm whose  competitive ratio is at most $3+2\sqrt{2}$. In case
 when we have one sender and several receivers, i.e., $n_s=1$ and $n_r>1$  our evacuation algorithm has competitive ratio at most $5$, and when $n_s > 1$ and $n_r=1$ we specify an algorithm with competitive ratio at most $\approx 5.681319$. These results can be found in Theorems~\ref{thm:1s_1r_ub}, \ref{thm:1s_2r_ub}, and \ref{thm:2s_1r_ub} respectively.

        In Section~\ref{sec:lower_bounds} we consider lower bounds on the competitive ratio of  evacuation algorithms for only the cases with only one sender agent, i.e., $n_s=1$. In particular, we prove a lower bound matching our upper bound for the case of $n_s=n_r=1$, which proves the optimality of our algorithm.  For the case of $n_s=1$ and $n_r>1$ we demonstrate that the evacuation cannot be completed with a competitive ratio less than $2+\sqrt{5}$. We conclude the paper with a discussion of open problems in Section~\ref{sec:conclusions}.
        
\section{Evacuation Algorithms and their Competitive Ratios}
\label{sec:upper_bounds}

In this section we give evacuation algorithms for our communication model and investigate their competitive ratios. We consider separately the cases, first  the single sender and single receiver, second the single sender and multiple receivers, and lastly the multiple senders and single receiver. 

        \subsection{One sender, one receiver}
                \begin{theorem}\label{thm:1s_1r_ub}
                        When $n_s=n_r=1$ there exists an evacuation algorithm with competitive ratio $3+2\sqrt{2}$.
                \end{theorem}
                \begin{proof}
                        The proof is constructive and based on the following algorithm: the receiver moves to the left at unit speed and the sender moves to the right with speed $\sqrt{2}-1$. If the sender finds the target first then it notifies the receiver (wirelessly) and the receiver moves at unit speed to the target. If the receiver finds the target first then it moves at unit speed to the right until it reaches the sender at which time both agents will move at full speed back to the target. We illustrate this algorithm in Figure~\ref{fig:1s1r_ub} using a space-time diagram which plots an agent's position on the $x$-axis, and uses the $y$-axis to indicate the flow of time.

                        \begin{figure}[!h]
                                \centering
                                \includegraphics[scale=0.1]{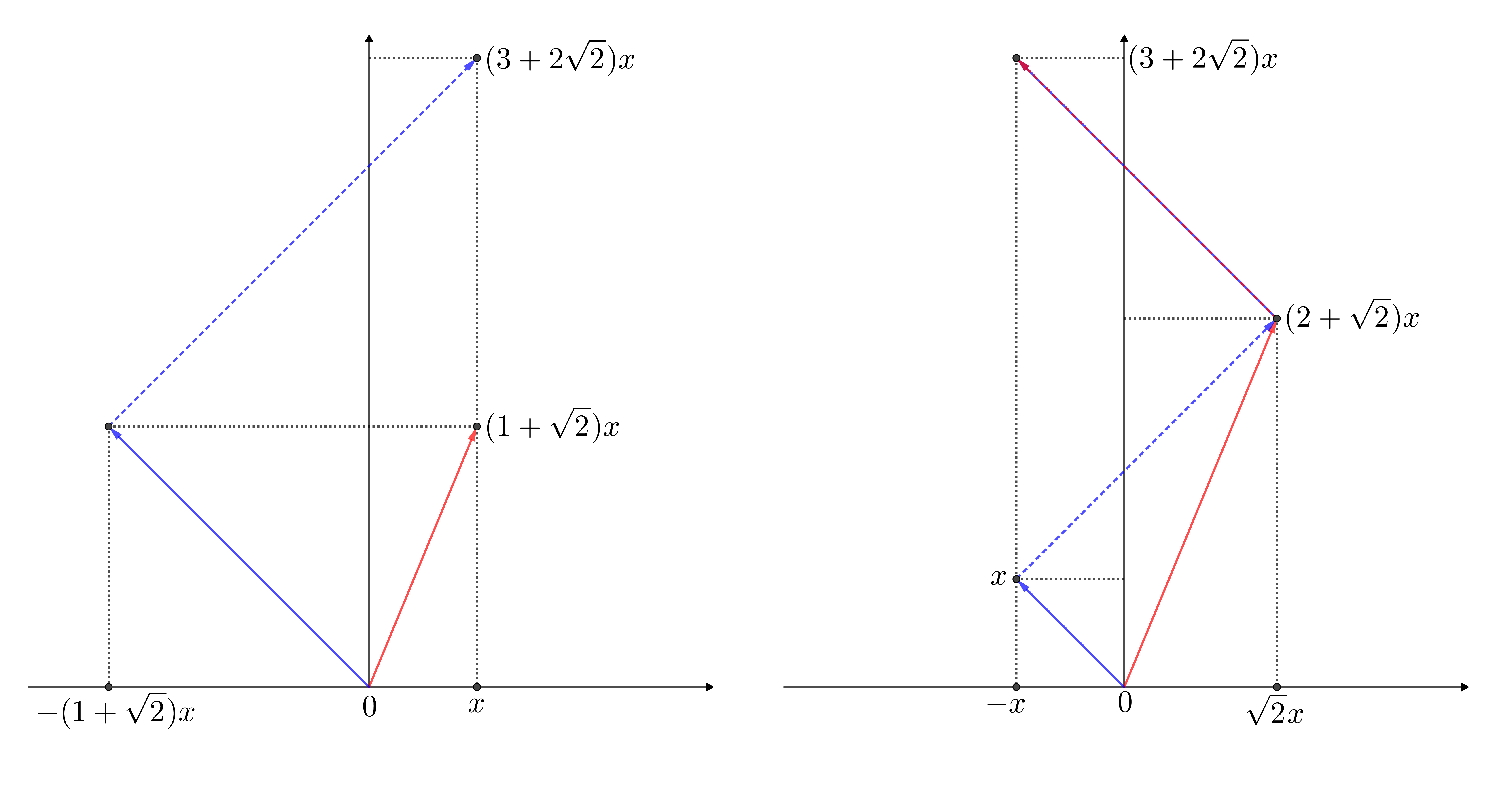}
                                \caption{The trajectories of the agents when the target is at location $+x$ (left) and $-x$ (right). The sender is colored red and the receiver is blue. A dashed line indicates when an agent deviates from its assigned search trajectory. Significant times and positions are indicated.\label{fig:1s1r_ub}}
                        \end{figure}

                        Suppose that the target is at location $x > 1$. The sender will find this target first and will do so at the time $\frac{x}{\sqrt{2}-1} = (1+\sqrt{2})x$. The sender immediately notifies the receiver which is at location $-(1+\sqrt{2})x$. Moving at unit speed, the receiver will travel distance $x + (1+\sqrt{2})x=(2+\sqrt{2})x$ to reach the target and will arrive at time $(1+\sqrt{2})x+(2+\sqrt{2})x = (3+2\sqrt{2})x$. The competitive ratio when $x>1$ is thus $3+2\sqrt{2}$.

                        Suppose now that the target is at location $-x < -1$. The receiver will find the target first and will do so at the time $x$. The receiver must move to notify the sender who is located at $(\sqrt{2}-1)x$ at time $x$. Hence, the distance between the agents is $x+(\sqrt{2}-1)x = \sqrt{2}x$ and the receiver will need to cross this distance with a relative speed of $1 - (\sqrt{2}-1) = 2-\sqrt{2}$. The receiver will thus take time $\frac{\sqrt{2}x}{2 - \sqrt{2}} = \frac{x}{\sqrt{2}-1} = (1+\sqrt{2})x$, and both agents will take an additional time $(1+\sqrt{2})x$ to reach the target. The time to evacuate is thus $x + 2(1+\sqrt{2})x = (3+2\sqrt{2})x$, and, evidently, the competitive ratio in this case is also $3+2\sqrt{2}$.
                \end{proof}
Notice that in the algorithm of Theorem \ref{thm:1s_1r_ub} it is essential that the sender moves initially at speed less than~1. 
                
\subsection{One sender, multiple receivers}
                \begin{theorem}\label{thm:1s_2r_ub}
                        When $n_s=1$ and $n_r > 1$ there exists an evacuation algorithm with competitive ratio $5$.
                \end{theorem}
                
               
                \begin{proof}
                        The proof is constructive and based on the following algorithm: one receiver moves to the left at unit speed and one receiver moves to the right at unit speed. When one of the receivers finds the target it immediately moves to notify the sender (and all other agents) at the origin. The sender immediately notifies the remaining receiver, and all agents proceed to the target at unit speed. An illustration of this algorithm is provided in Figure~\ref{fig:1s2r_ub} for the case that the target is at location $-x < -1$. The situation is symmetric when $x > 1$.
                        
                        \begin{figure}[!h]
                                \centering
                                \includegraphics[scale=0.1]{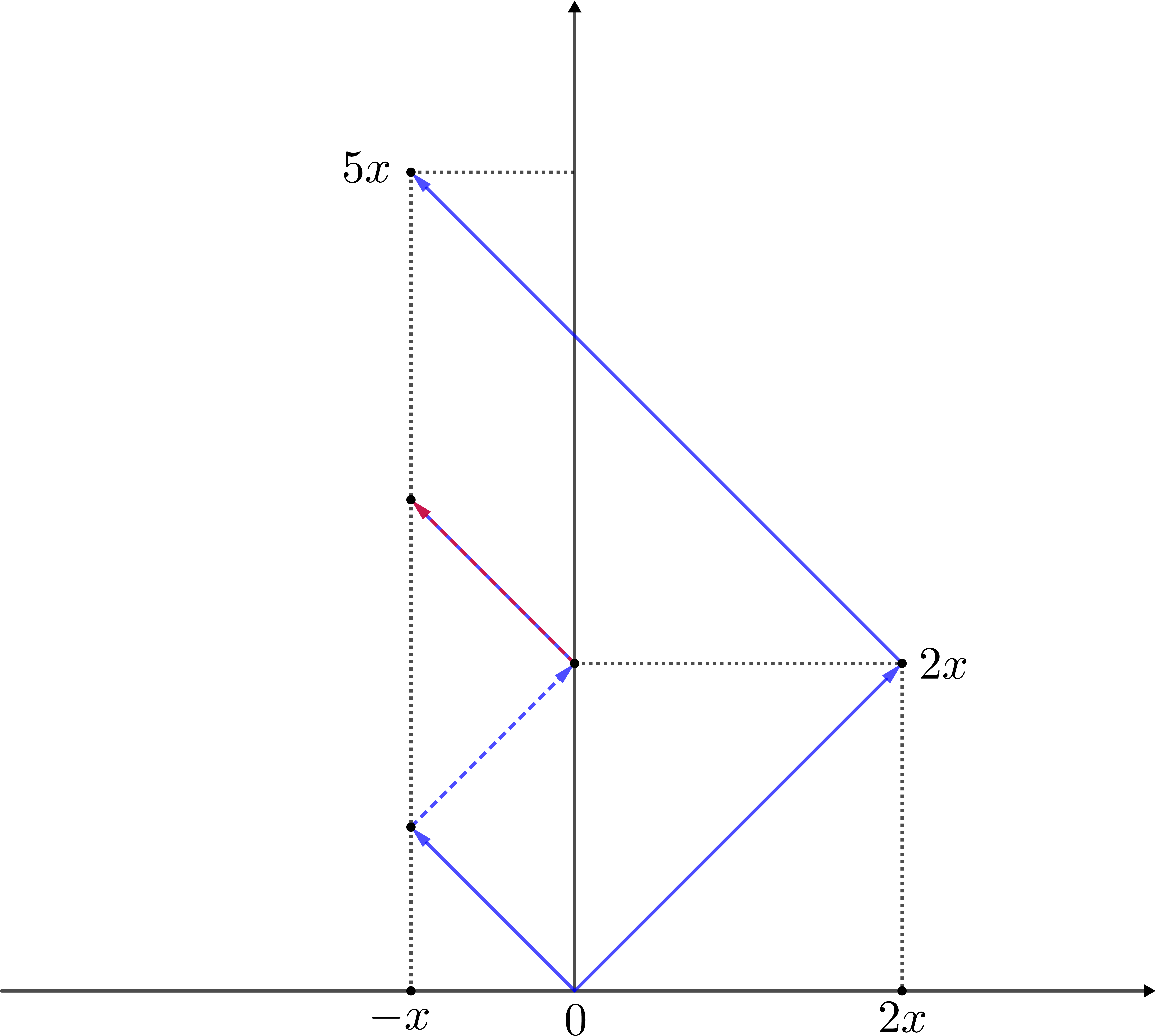}
                                \caption{The trajectories of the agents when the target is at location $-x$. The sender is in red and the receivers are in blue. A dashed line indicates when an agent deviates from its assigned search trajectory. Significant times and positions are indicated.\label{fig:1s2r_ub}}
                        \end{figure}

                        Suppose that the target is at location $-x < -1$. The left receiver will find this target first and will do so at the time $x$. It then immediately moves to the origin to notify the sender, arriving at time $2x$. The sender notifies the right receiver who is at location $2x$. The right receiver then moves at unit speed to the target arriving at time $5x$. The competitive ratio is thus $5$. The case when $x > 1$ is totally symmetric and also yields a competitive ratio of 5.
                \end{proof}

        \subsection{Multiple senders, one receiver}
                \begin{theorem}\label{thm:2s_1r_ub}
                        When $n_s > 2$ and $n_r = 1$ there exists an algorithm $\Alg$ with competitive ratio $\CR(\Alg) < 5.681319$. More exactly, the competitive ratio is upper bounded by
                        \begin{equation}
                                \CR(\Alg) \le 1 + \frac{1+v_r}{1-v_r}\left(\frac{1+4v_r-v_r^2}{v_r(3-v_r)}\right)
                        \end{equation}
                        with $v_r$ chosen to be the root of the equation $v_r^4 - 16v_r^3 + 26v_r^2 + 8v_r - 3 = 0$ satisfying $0 \leq v_r < 1$.
                \end{theorem}                
                \vspace{-0.3cm} The proof of this result is much more involved than the previous two cases. 
                 When there are more than two senders, all but two of the senders will remain at the origin until they are notified of the target (at which time they move to the announced location). Thus, in the rest of this section we will present our algorithm for the specific case of two senders and one receiver, i.e., $n_s=2$ and $n_r=1$. 

\begin{figure}[!htb]
  \begin{center}
    \includegraphics[scale=0.1]{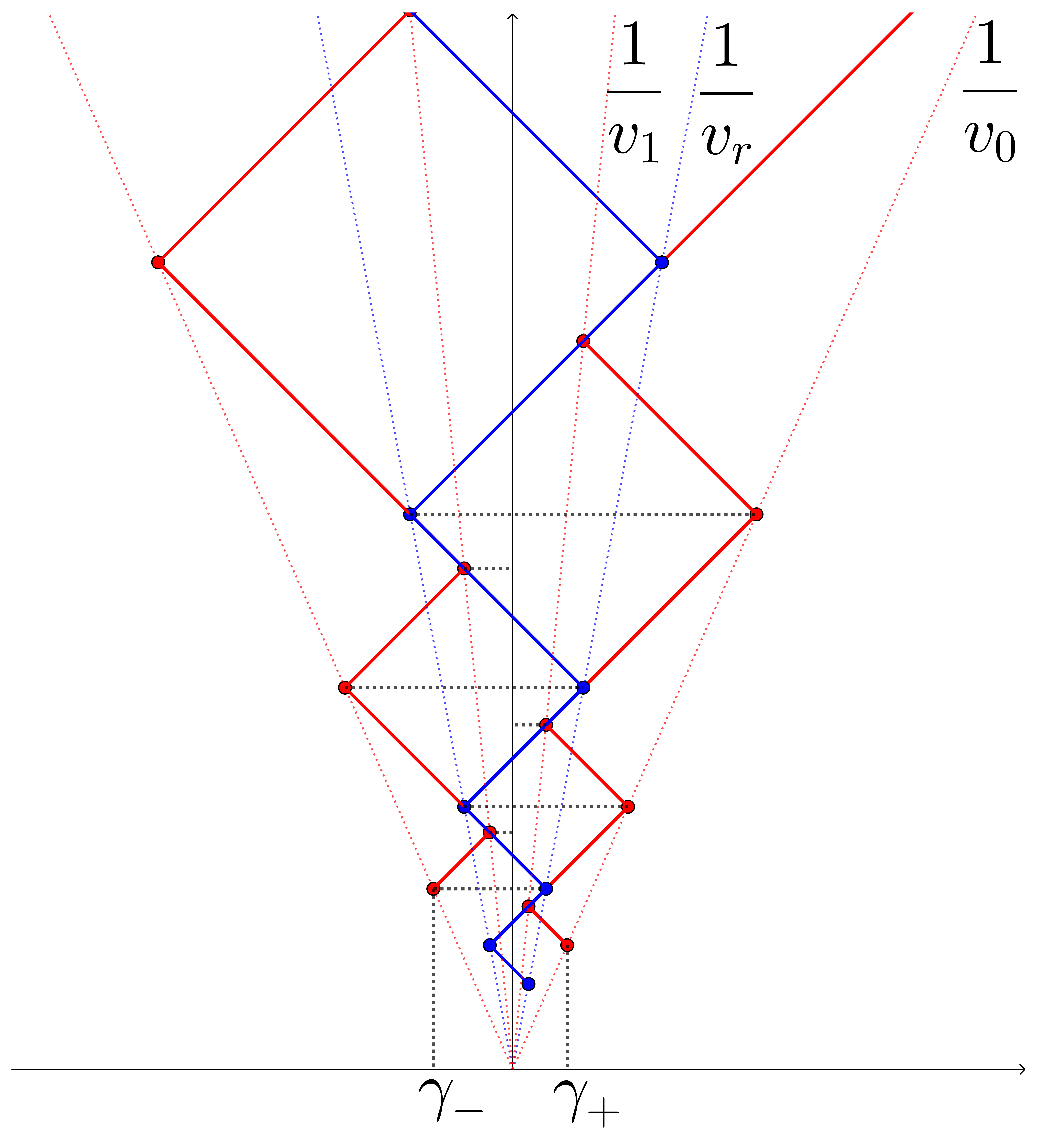}
  \end{center}
  \caption{Trajectories of the agents for the evacuation algorithm $\textsc{EVAC}_\textsc{Rays}(v_r)$. The sender trajectories are red, and the receiver trajectory is blue. \label{fig:2s1r_ub}}
\end{figure}

\emph{High level idea.} The robots jointly maintain an interval around the origin of positions that have already been explored by at least one robot. They wish to expand this interval at a fast pace while maintaining the ability to notify all robots quickly in case an exit is found by at least one robot. The idea behind our algorithm is to make one sender responsible for extending the right end of the searched interval, and another sender responsible for extending the left end of the searched interval. The receiver zig-zags around the origin (with the lengths of zigs and zags increasing in rounds), so that if one of the senders finds an exit, the receiver is ``close'' to the other sender and can quickly notify it via F2F. In order for this idea to work, the senders cannot simply move away from the origin at full speed, but instead they perform zig-zags of their own (however, unlike the receiver their zig-zags are drifting away from the origin). One can think of a sender as first extending the searched region for a while and then coming back partway towards the origin to get notified by the receiver about what's happening on the other side of the origin. This strategy is illustrated in Figure~\ref{fig:2s1r_ub}. When the exit is found by one of the senders, the receiver goes to intercept the other sender and they both move towards the exit.

An interesting feature of the algorithm is that the zig-zag trajectory of the receiver non trivially overlaps with the zig-zag trajectories of the senders, i.e., it does not simply touch them. For example, take a particular time when the receiver meets the right sender for the first time during one zig-zag round. Then the receiver and the right sender travel to the right together for some time. During this time the left sender extends the searched region on the left. If an exit is found by the left sender at this point, this is good -- both the right sender and receiver will learn about it instantaneously and will start moving towards the exit. However, the right sender and receiver cannot keep travelling together for very long, since the trajectory needs to have certain symmetries, lest the left sender gets too far. Thus, at some point the receiver and the right sender part ways with the receiver moving towards the left sender and the right sender continuing to the right. At precisely this point, the left sender stops extending the search interval and starts to move towards the receiver (this situation is indicated by dashed lines in Figure~\ref{fig:2s1r_ub}). Intuitively, this is a good timing for the left sender to switch direction, because otherwise if it finds an exit soon after the receiver and right sender part ways then the receiver would not be able to catch up with the right sender for quite a while (until the right sender's next ``zag'').


Formalizing and analyzing this algorithm takes a lot of work and careful calculations. 
    We begin by introducing a class of search trajectories that are parameterized by a four-tuple $[\eta,v_0,v_1,\gamma]$ where: $\eta=\pm 1$, and $v_0$, $v_1$, and $\gamma$ are real numbers satisfying $0 \leq v_0 \leq 1$, $-1 \leq v_1 < v_0$, and $0 < \gamma \leq 1$.

                {
                \vspace{-0.3cm}
                        \makeatletter
                        \renewcommand{\ALG@name}{Trajectory}
                        \makeatother

                \begin{algorithm}[!h] \caption{$\textsc{Rays}(\eta,v_0,v_1,\gamma)$} \label{alg:2s1r_traj}
                        \begin{algorithmic}[1]             
                                \Begin Move to location $\eta \gamma$ and wait until time $\frac{\gamma}{v_0}$.
                                \MRepeat
                                \State Move in direction $-\eta$ at unit speed until my position $x$ at time $t$ satisfies $\frac{x}{t} = \eta v_1$;
                                \State Move in direction $\eta$ at unit speed until my position $x$ at time $t$ satisfies $\frac{x}{t} = \eta v_0$;
                                \EndRepeat
                                \End
                        \end{algorithmic}
                \end{algorithm}  
                }

\begin{figure}[!htb]
  \begin{center}
    \includegraphics[scale=0.15]{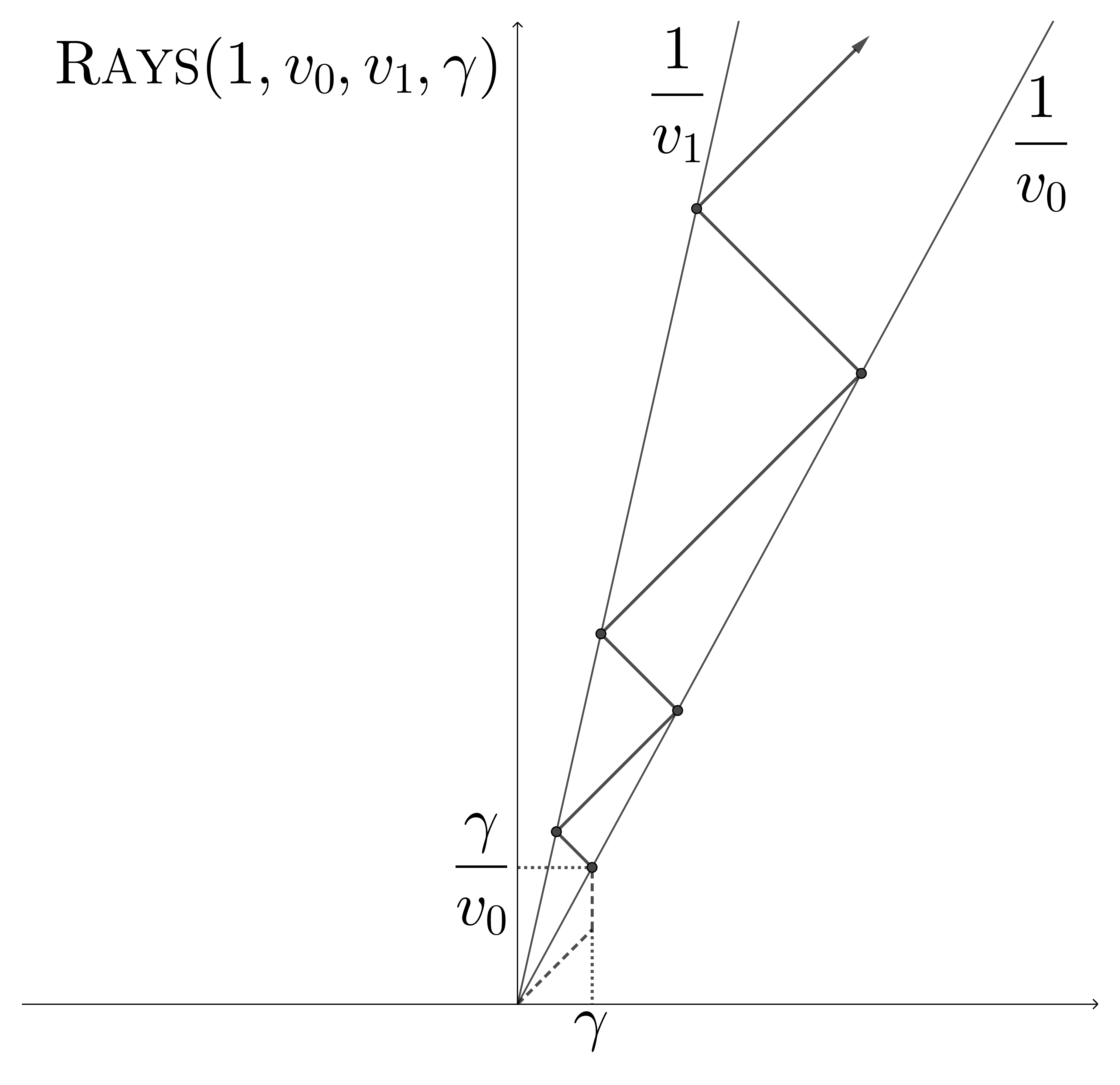}
  \end{center}
  \caption{Example of the trajectory $\textsc{Rays}(1,v_0,v_1,\gamma)$. \label{fig:cone_alg}}
\end{figure}
                An example of this type of search trajectory is illustrated in Figure~\ref{fig:cone_alg}. As one can observe, after an initial setup phase, the trajectory $\textsc{Rays}(\eta,v_0,v_1,\gamma)$ will bounce back and forth between two space-time rays with slopes $\frac{\eta}{v_0}$ and $\frac{\eta}{v_1}$. The parameter $\gamma$ dictates the ``beginning'' position on the ray with slope $\frac{\eta}{v_0}$, and $\eta$ is a symmetry parameter in the sense that trajectories $\textsc{Rays}(\pm 1,v_0,v_1,\gamma)$ are reflections of each other about the time-axis. Although the above specification of the trajectory $\textsc{Rays}(\eta,v_0,v_1,\gamma)$ is simple to understand, it will be more convenient to express these trajectories in terms of their {\em turning-points} -- space-time points at which an agent changes its travel direction, and between which an agent moves at constant unit speed. One can observe from Figure~\ref{fig:cone_alg} that the turning-points of $\textsc{Rays}(\eta,v_0,v_1,\gamma)$ are precisely the points where the trajectory bounces off the rays with slopes $\frac{\eta}{v_0}$ and $\frac{\eta}{v_1}$. The next lemma provides expressions for the turning-points of the trajectory $\textsc{Rays}(\eta,v_0,v_1,\gamma)$.
%
                
                
                \begin{lemma}\label{lm:tp_general}
                        The turning-points $(D_j,T_j)$, $j=0,1,\ldots$, of the trajectory $\textsc{Rays}(\eta,v_0,v_1,\gamma)$ are given by
                        \[D_j = \eta \gamma \left[\frac{(1-v_1)(1+v_0)}{(1+v_1)(1-v_0)} \right]^{\floor{\frac{j}{2}}}\begin{cases}
                                1,& \mbox{even }j\\
                                \frac{v_1(1+v_0)}{v_0(1+v_1)},& {\mbox{odd }j}
                        \end{cases},\qquad T_j = \frac{D_j}{\eta}\begin{cases}
                                \frac{1}{v_0},& \mbox{even }j\\
                                \frac{1}{v_1},& \mbox{odd }j.
                        \end{cases}\]
                \end{lemma}
                
                \begin{proof}
                        We will derive the turning-points when $\eta=1$. The turning-points for $\eta=-1$ are just reflections about the time axis. 
                        
                        The first turning-point is $P_0 = (\gamma,\gamma/v_0)$ which is evident from the description of $\textsc{Rays}(1,v_0,v_1,\gamma)$. This turning-point lies on the ray of slope $\frac{1}{v_0}$ and so the next turning-point will lie on the ray with slope $\frac{1}{v_1}$. The turning-points for larger $j$ will then alternate between these two rays. It follows that the turning-times $T_j$ can be expressed in terms of the turning-positions $D_j$ as follows:
                        \[T_j = D_j\begin{cases}
                                \frac{1}{v_0},& \mbox{even }j\\
                                \frac{1}{v_1},& \mbox{odd }j.
                        \end{cases}\]
                        We can therefore focus on finding the turning positions $D_j$. 

                        The agents travel at unit speed between turning-points $P_{j-1}$ and $P_{j}$, and an agent will be moving to the right/left when $j$ is even/odd. When $j$ is even we have
                        \[1 = \frac{T_{j}-T_{j-1}}{D_{j}-D_{j-1}} = \frac{\frac{D_j}{v_0}-\frac{D_{j-1}}{v_1}}{D_{j}-D_{j-1}} \quad\rightarrow\quad \left(\frac{1}{v_0}-1\right)D_{j} = \left(\frac{1}{v_1}-1\right)D_{j-1}\]
                        and finally
                        \[D_j = \frac{v_0(1-v_1)}{v_1(1-v_0)}D_{j-1},\ \mbox{even j}.\]
                        When $j$ is odd we find in a similar manner that
                        \[D_j = \frac{v_1(1+v_0)}{v_0(1+v_1)}D_{j-1},\ \mbox{odd j}.\]
                        Combining these results yields, for even or odd $j$,
                        \[D_j = \frac{(1-v_1)(1+v_0)}{(1+v_1)(1-v_0)} D_{j-2}.\]
                        Unrolling this recursion then gives
                        \[D_j = \left[\frac{(1-v_1)(1+v_0)}{(1+v_1)(1-v_0)} \right]^{\floor{\frac{j}{2}}} \begin{cases}
                                D_0,& \mbox{even j}\\
                                D_1,& \mbox{odd j}
                        \end{cases} = \gamma \left[\frac{(1-v_1)(1+v_0)}{(1+v_1)(1-v_0)} \right]^{\floor{\frac{j}{2}}} \begin{cases}
                                1,& \mbox{even j}\\
                                \frac{v_1(1+v_0)}{v_0(1+v_1)},& \mbox{odd j}
                        \end{cases}\]
                        where we have used the fact that $D_0 = \gamma$, and our expression for $D_j$ when $j$ is odd.
                \end{proof}
                

                We will describe our evacuation algorithm in terms of the trajectories $\textsc{Rays}(1,v_0,v_1,\gamma)$. To this end, we represent by $X_\pm(t)$ the trajectories of the senders and we refer to the sender with trajectory $X_+$ (resp. $X_-$) as the {\em right-sender} (resp. {\em left-sender}). We use $X_r(t)$ to represent the trajectory of the receiver. The turning-points of the trajectories $X_\pm$ will be represented by $(D^\pm_j,T^\pm_j)$, and the turning-points of the trajectory $X_r$ will be represented by $(D^r_j, T^r_j)$. With this notation our evacuation algorithm can be expressed as in Algorithm~\ref{alg:2s1r_ub}. We refer to this algorithm by $\textsc{Evac}_\textsc{Rays}(v_r)$.

                \begin{algorithm}[!] \caption{$\textsc{Evac}_\textsc{Rays}(v_r)$, $0 \leq v_r < 1$} \label{alg:2s1r_ub}
                        \begin{algorithmic}[1]             
                                \State \begin{equation}\label{eq:2s1r_alg}
                                        X_\pm(t) = \textsc{Rays}(\pm 1,v_0,v_1,\gamma_\pm),\quad X_r(t) = \textsc{Rays}(1,v_r,-v_r,1).
                                \end{equation}
                                \begin{equation}\label{eq:2s1r_pars}
                                        v_0 = \frac{v_r(3-v_r)}{1+v_r},\quad v_1 = \frac{v_r(1-v_r)}{1+3v_r},\quad \gamma_+ = \frac{3-v_r}{1-v_r},\quad \gamma_- = \frac{3-v_r}{1-v_r} \frac{1+v_r}{1-v_r}
                                \end{equation}                                
                        \end{algorithmic}
                \end{algorithm} 

                Figure~\ref{fig:2s1r_ub} illustrates the trajectories of the agents for the algorithm $\textsc{Evac}_\textsc{Rays}(v_r)$. The choices of $v_0,v_1,\gamma_\pm$ in \eqref{eq:2s1r_pars} ensure that the trajectories enjoy a number of important properties, some of which are evident in Figure~\ref{fig:2s1r_ub}. One immediately obvious property is the fact that the right/left-sender spends all of its time to the right/left of the origin (and hence the naming convention). Some other properties that are evident in Figure~\ref{fig:2s1r_ub} are given in Observation~\ref{obs:2s1r_alg}.
                \begin{observation}\label{obs:2s1r_alg}
                        For all $k=0,1,2,\ldots$ the following properties hold for the algorithm $\textsc{Evac}_\textsc{Rays}(v_r)$:
                        \begin{enumerate}
                                \item the receiver reaches its turning point $2k+1$ (resp. $2k+2$) at the same time the right-sender (resp. left-sender) reaches its turning-point $2k$.
                                \item the receiver and right-sender (resp. left-sender) are co-located at all times in the interval $[T^+_{2k+1},T^r_{2k+2}]$ (resp. $[T^-_{2k+1},T^r_{2k+3}]$),
                        \end{enumerate}
                \end{observation}
%

                In order to establish these properties, we carefully calculate the turning points of all agents in terms of the parameter $v_r$. The following lemmas summarize the calculations. 
                Equipped with these formulas, Observation~\ref{obs:2s1r_alg} follows.
                \begin{lemma}\label{lm:tp_alg_r}
                        The turning-points of the receiver are
                        \[D^r_j = (-1)^j \left(\frac{1+v_r}{1-v_r} \right)^{j},\qquad T^r_j = \frac{1}{v_r}\left(\frac{1+v_r}{1-v_r} \right)^{j}.\]
                \end{lemma}
                
                \begin{proof}
                        With $X_r(t) = \textsc{Rays}(1,-v_r,v_r,1)$ it follows from Lemma~\ref{lm:tp_general} that for even $j$ we have
                        \begin{align*}
                                D^r_j &= \left[\frac{(1+v_r)(1+v_r)}{(1-v_r)(1-v_r)} \right]^{\floor{\frac{j}{2}}} = \left( \frac{1+v_r}{1-v_r} \right)^{2\floor{\frac{j}{2}}} = \left( \frac{1+v_r}{1-v_r} \right)^j.
                        \end{align*} 
                        and for odd $j$ we similarly find that $D^r_j = -\left( \frac{1+v_r}{1-v_r} \right)^{j}$. Thus, for $j$ even or odd we have $D^r_j = (-1)^j\left(\frac{1+v_r}{1-v_r} \right)^{j}$. The times $T^r_j$ are
                        \[T^r_j = D^r_j \begin{cases}
                                \frac{1}{v_r},& \mbox{even }j\\
                                -\frac{1}{v_r},& \mbox{odd }j.
                        \end{cases} = (-1)^j \frac{D^r_j}{v_r} = \frac{1}{v_r}\left(\frac{1+v_r}{1-v_r} \right)^{j}.\]
                \end{proof}

                We note the following identities concerning the turning-points $(D^r_j,T^r_j)$ which we will use these identities often and without reference.
                \[T^r_j = \frac{1+v_r}{1-v_r}T^r_{j-1} = \frac{1-v_r}{1+v_r}T^r_{j+1},\qquad D^r_j = -\frac{1+v_r}{1-v_r}D^r_{j-1} = -\frac{1-v_r}{1+v_r}D^r_{j+1}.\]


                \begin{lemma}\label{lm:tp_alg_right}
                        The turning-points of the right-sender are
                        \[D^+_j =v_r T^r_j \begin{cases}
                                \frac{3-v_r}{1-v_r},& \mbox{even }j\\
                                \frac{1-v_r}{1+v_r},& {\mbox{odd }j}
                        \end{cases},\qquad T^+_j = T^r_j\begin{cases}
                                \frac{1+v_r}{1-v_r},& \mbox{even }j\\
                                \frac{1+3v_r}{1+v_r},& {\mbox{odd }j}
                        \end{cases}\]
                \end{lemma}
                
                \begin{proof}
                        With $X_+(t) = \textsc{Rays}(1,v_0,v_1,\gamma_+)$ it follows from Lemma~\ref{lm:tp_general} that
                        \[D^+_j = \gamma_+ \left[\frac{(1-v_1)(1+v_0)}{(1+v_1)(1-v_0)} \right]^{\floor{\frac{j}{2}}}\begin{cases}
                                1,& \mbox{even }j\\
                                \frac{v_1(1+v_0)}{v_0(1+v_1)},& {\mbox{odd }j}
                        \end{cases},\qquad T^+_j = D^+_j\begin{cases}
                                \frac{1}{v_0},& \mbox{even }j\\
                                \frac{1}{v_1},& \mbox{odd }j.
                        \end{cases}\]    
                        With $v_0$ and $v_1$ given by \eqref{eq:2s1r_pars} we have
                        \begin{align*}
                                \frac{(1-v_1)(1+v_0)}{(1+v_1)(1-v_0)} &= \frac{\left(1-\frac{v_r(1-v_r)}{1+3v_r}\right)\left(1+\frac{v_r(3-v_r)}{1+v_r}\right)}{\left(1+\frac{v_r(1-v_r)}{1+3v_r}\right)\left(1-\frac{v_r(3-v_r)}{1+v_r}\right)}\\
                                &= \frac{\left(1+2v_r+v_r^2\right)\left(1+3v_r+v_r(1-v_r)\right)}{\left(1+3v_r+v_r(1-v_r)\right)\left(1-2v_r+v_r^2\right)} = \left(\frac{1+v_r}{1-v_r}\right)^2
                        \end{align*}
                        We also observe that
                        \begin{align*}
                                \frac{v_1(1+v_0)}{v_0(1+v_1)} &= \frac{\frac{v_r(1-v_r)}{1+3v_r}\left(1+\frac{v_r(3-v_r)}{1+v_r}\right)}{\frac{v_r(3-v_r)}{1+v_r}\left(1+\frac{v_r(1-v_r)}{1+3v_r}\right)} \\ &= \frac{v_r(1-v_r)(1+v_r+v_r(3-v_r)}{v_r(3-v_r)(1+3v_r+v_r(1-v_r))} = \frac{1-v_r}{3-v_r} = \frac{1}{\gamma_+}.
                        \end{align*}
                        Substituting these last two results into our expression for $D^+_j$ then yields
                        \[D^+_j = \left(\frac{1+v_r}{1-v_r}\right)^{2\floor{\frac{j}{2}}}\begin{cases}
                                \gamma_+,& \mbox{even }j\\
                                1,& {\mbox{odd }j}
                        \end{cases} = \begin{cases}
                                \gamma_+\left(\frac{1+v_r}{1-v_r}\right)^{j},& \mbox{even }j\\
                                \left(\frac{1+v_r}{1-v_r}\right)^{j-1},& {\mbox{odd }j}
                        \end{cases} = v_r T^r_j \begin{cases}
                                \frac{3-v_r}{1-v_r},& \mbox{even }j\\
                                \frac{1-v_r}{1+v_r},& {\mbox{odd }j}
                        \end{cases}\]
                        as required. 
                        
                        For $T^+_j$ we have
                        \[T^+_j = D^+_j\begin{cases}
                                \frac{1}{v_0},& \mbox{even }j\\
                                \frac{1}{v_1},& \mbox{odd }j.
                        \end{cases} = v_r \begin{cases}
                                \frac{\gamma_+}{v_0} T^r_{j},& \mbox{even }j\\
                                \frac{1}{v_1}T^r_{j-1},& {\mbox{odd }j}
                        \end{cases}.\]
                        We observe that
                        \[\frac{\gamma_+}{v_0} = \frac{\frac{3-v_r}{1-v_r}}{\frac{v_r(3-v_r)}{1+v_r}} = \frac{1}{v_r}\left(\frac{1+v_r}{1-v_r}\right)\]
                        and thus
                        \[T^+_j = v_r\begin{cases}
                                \frac{1}{v_r}\left(\frac{1+v_r}{1-v_r}\right)T^r_{j},& \mbox{even }j\\
                                \frac{1+3v_r}{v_r(1-v_r)}T^r_{j-1},& {\mbox{odd }j}
                        \end{cases} = \begin{cases}
                                T^r_{j+1},& \mbox{even }j\\
                                \frac{1+3v_r}{1-v_r}T^r_{j-1},& {\mbox{odd }j}
                        \end{cases} = T^r_j\begin{cases}
                                \frac{1+v_r}{1-v_r},& \mbox{even }j\\
                                \frac{1+3v_r}{1+v_r},& {\mbox{odd }j}
                        \end{cases}.\]
                        This completes the proof.
                \end{proof}
                

                \begin{lemma}\label{lm:tp_alg_left}
                        The turning-points of the left-sender are
                        \[D^-_j = -v_r T^r_{j+1} \begin{cases}
                                \frac{3-v_r}{1-v_r},& \mbox{even }j\\
                                \frac{1-v_r}{1+v_r},& {\mbox{odd }j}
                        \end{cases},\qquad T^-_j = T^r_{j+1}\begin{cases}
                                \frac{1+v_r}{1-v_r},& \mbox{even }j\\
                                \frac{1+3v_r}{1+v_r},& {\mbox{odd }j}
                        \end{cases}.\]      
                \end{lemma}
                \begin{proof}
                        The proof is essentially identical to the proof of Lemma~\ref{lm:tp_alg_right}.
                \end{proof}

                We are now ready to prove our previous observations about this algorithm.
                \begin{proof}[Proof of Observation~\ref{obs:2s1r_alg}] We will prove the properties for the right-sender only. Those for the left-sender follow in a nearly identical manner.
                        
                        The first statement we want to prove is: ``the receiver reaches its turning point $2k+1$ at the same time the right-sender reaches its turning-point $2k$''. The receiver reaches its turning-point $2k+1$ at time $T^r_{2k+1}$. The right-sender reaches its turning-point $2k$ at time $T^+_{2k}$ and by Lemma~\ref{lm:tp_alg_right} we have  $T^+_{2k} = \frac{1+v_r}{1-v_r}T^r_{2k} = T^r_{2k+1}$, which proves the statement.

                        The second statement we want to prove is: ``the receiver and right-sender are co-located at all times in the interval $[T^+_{2k+1},T^r_{2k+2}]$''. During the interval $[T^+_{2k+1},T^+_{2k+2}]$ the right-sender will be moving to the right along the space-time line
                        \[t = x-D^+_{2k+1}+T^+_{2k+1} = x - v_r\frac{1-v_r}{1+v_r}T^r_{2k+1} + \frac{1+3v_r}{1+v_r}T^r_{2k+1} = x + \frac{1+2v_r+v_r^2}{1+v_r}T^r_{2k+1}\]
                        and finally
                        \begin{equation}\label{eq:2s1r_right_move}
                                t = x + (1+v_r)T^r_{2k+1}.
                        \end{equation}
                        During the interval $[T^r_{2k+1},T^r_{2k+2}]$ the receiver will be moving to the right along the space-time line
                        \[t = x - D^r_{2k+2}+T^r_{2k+2} = x - v_r T_R(k+2)+T_R^{2k+2} = x + (1-v_r)T^r_{2k+2} = x + (1+v_r)T^r_{2k+1}.\]
                        We can thus conclude that the right-sender and receiver will be travelling along the same space-time line and will be co-located during the interval $[T^+_{2k+1},T^+_{2k+2}] \cap [T^r_{2k+1},T^r_{2k+2}] = [T^+_{2k+1},T^r_{2k+2}]$.
                \end{proof}

                The next theorem provides an expression for the competitive ratio of $\textsc{Evac}_\textsc{Rays}(v_r)$ as a function of $v_r$. 
                \begin{theorem}\label{thm:2s1r_cr}
                        The competitive ratio of algorithm $\textsc{Evac}_\textsc{Rays}(v_r)$ satisfies
                        \[\CR \le 1 + \frac{1+v_r}{1-v_r}\left(\frac{1+4v_r-v_r^2}{v_r(3-v_r)}\right).\]
                \end{theorem}
                
                \begin{proof}
                        Due to the symmetry between the right/left-senders, we may assume without loss of generality that the target is found by the right-sender. Moreover, the sequence of intervals $(D^+_{2k},D^+_{2k+2}]$, $k=0,1,2,\ldots$, collectively covers the entire line extending from $D^+_0$ to $+\infty$ and so we may assume without loss of generality that the target is at location $x_* \in (D^+_{2k},D^+_{2k+2}]$, for some fixed value of $k \geq 0$.
                        
                        \begin{figure}
                                \centering
                                \includegraphics[scale=0.15]{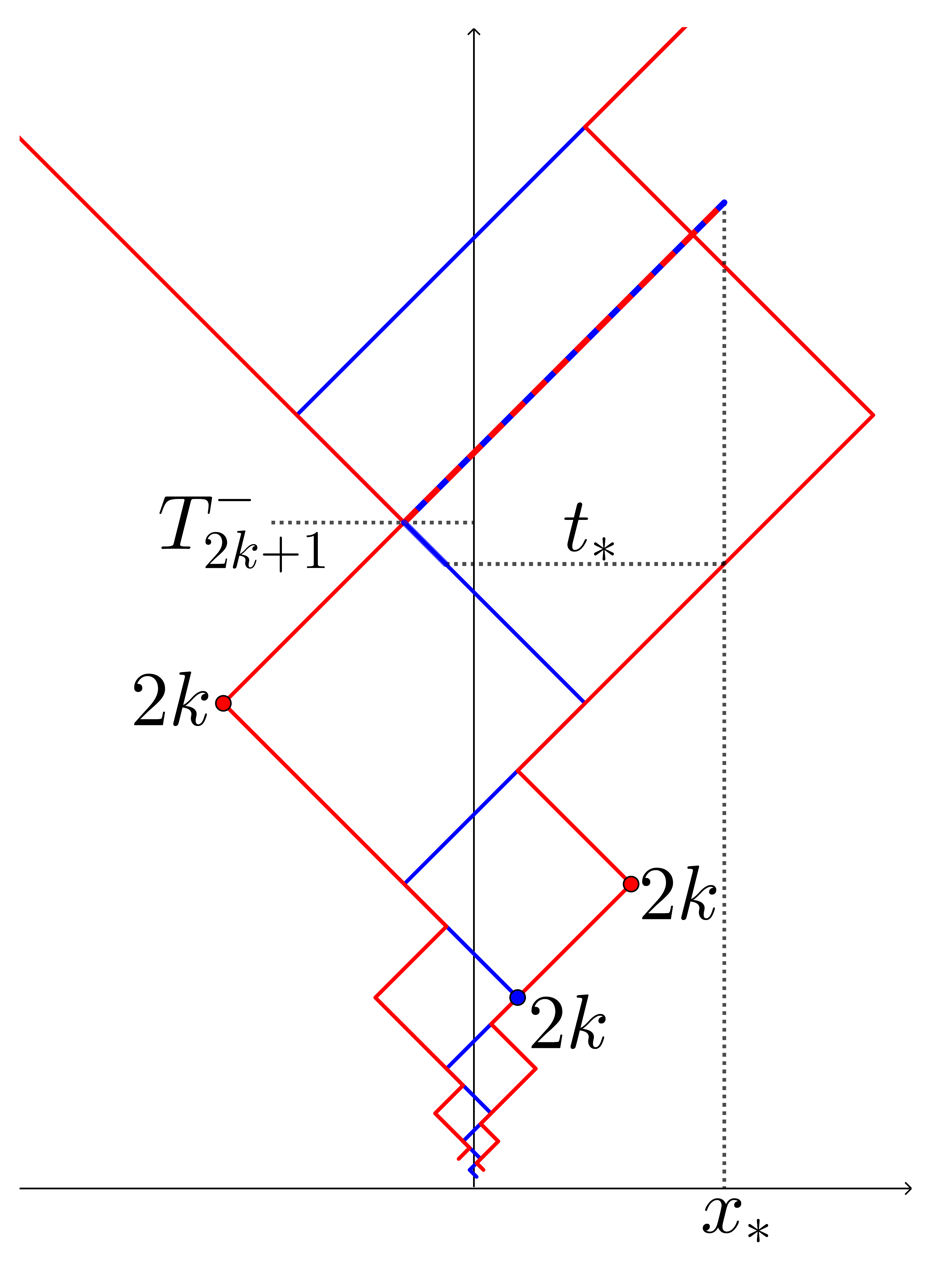}
                                \includegraphics[scale=0.15]{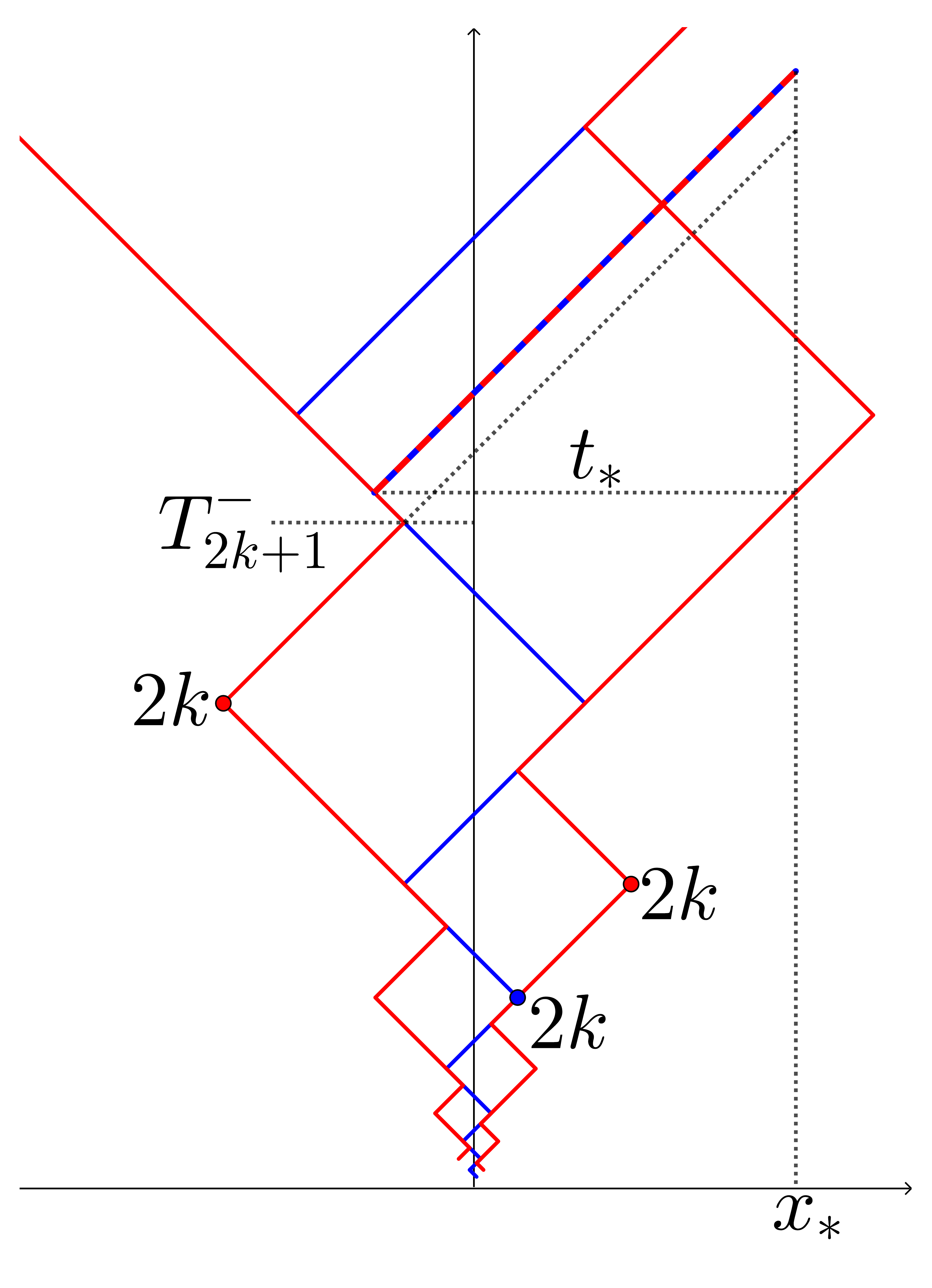}
                                \caption{Setup for the proof of Theorem~\ref{thm:2s1r_cr}. The turning-point $2k$ of each agent is indicated. On the left the target is found at time $t_* \leq T^-_{2k+1}$ and on the right the target is found at time $t_* > T^-_{2k+1}$.\label{fig:2s1r_cr}}
                        \end{figure}

                        The right-sender will reach $x_*$ while travelling to the right between its turning points $2k+1$ and $2k+2$, and, we demonstrated in the proof of Observation~\ref{obs:2s1r_alg} that while doing so this sender will be moving along the space-time line with equation \eqref{eq:2s1r_right_move}. Thus, the time $t_*$ at which the right-sender reaches the target is
                        \begin{align*}
                                t_* &= x_* + (1+v_r)T^r_{2k+1}.
                        \end{align*}
                        After reaching the target the right-sender will wirelessly notify the receiver and the receiver will move to notify the left-sender. There are two cases to consider, each of which is illustrated in Figure~\ref{fig:2s1r_cr}. In the first case -- left side of Figure~\ref{fig:2s1r_cr} -- the target is found at location $x_*$ such that $t_* \leq T^-_{2k+1}$. We know from Observation~\ref{obs:2s1r_alg} that the receiver will be co-located with the left-receiver at all times within the interval $[T^-_{2k+1},T^r_{2k+3}]$, and before time $T^-_{2k+1}$ the receiver and left-sender will be moving towards each other, each at unit speed. Thus, the earliest time that the left-sender could be notified of the target is at the time $T^-_{2k+1}$. Evidently, the evacuation time for this case is
                        \begin{align*}
                                E &= T^-_{2k+1} + |x_* - D^-_{2k+1}| = x_* + \frac{1+3v_r}{1+v_r} T^r_{2k+2} + v_r \frac{1-v_r}{1+v_r}T^r_{2k+2}\\
                                &= x_* + \frac{1+3v_r+v_r(1-v_r)}{1+v_r} T^r_{2k+2}.
                        \end{align*}
                        and the competitive ratio is
                        \begin{align*}
                                \CR = \frac{E}{x_*} &= 1 + \frac{\frac{1+3v_r+v_r(1-v_r)}{1+v_r} T^r_{2k+2}}{x_*}.
                        \end{align*}
                        The competitive ratio increases with decreasing $x_*$, and with $x_* > D^+_{2k} = v_r \frac{3-v_r}{1-v_r} T^r_{2k}$ we get
                        \begin{align*}
                                \CR &\le 1 + \frac{1+3v_r+v_r(1-v_r)}{1+v_r}\frac{T^r_{2k+2}}{v_r \frac{3-v_r}{1-v_r} T^r_{2k}}\\
                                &= 1 + \frac{1+3v_r+v_r(1-v_r)}{1+v_r} \cdot \frac{1-v_r}{v_r(3-v_r)} \cdot \frac{(1+v_r)^2}{(1-v_r)^2}\\
                                &= 1 + \frac{1+v_r}{1-v_r}\left(\frac{1+3v_r+v_r(1-v_r)}{v_r(3-v_r)}\right)
                                = 1 + \frac{1+v_r}{1-v_r}\left(\frac{1+4v_r-v_r^2}{v_r(3-v_r)}\right).
                        \end{align*}
                        and finally
                        \begin{equation}\label{eq:2s1r_cr}
                                \CR \le 1 + \frac{1+v_r}{1-v_r}\left(\frac{1+4v_r-v_r^2}{v_r(3-v_r)}\right).
                        \end{equation}

                        The second case -- the right side of Figure~\ref{fig:2s1r_cr} -- occurs when the target is found at a time $t_* \in (T^-_{2k+1},T^+_{2k+2}] = (T^-_{2k+1},T^r_{2k+3}]$. The left-sender and receiver are co-located during the time interval $(T^-_{2k+1},T^r_{2k+3}]$, and so the left-sender will be notified of the target at time $t_*$. By referring to Figure~\ref{fig:2s1r_cr} one can observe that the evacuation time for this case will be $2(t_*-T^-_{2k+1})$ more than the evacuation time of the previous case, i.e.,
                        \[E = x_* + \frac{1+3v_r+v_r(1-v_r)}{1+v_r} T^r_{2k+2} + 2(t_*-T^-_{2k+1}).\]
                        Since $t_* = x_* + (1+v_r)T^r_{2k+1} = x_* + (1-v_r)T^r_{2k+2}$ and $T^-_{2k+1} = \frac{1+3v_r}{1+v_r}T^r_{2k+2}$ we get
                        \begin{align*}
                                E &= 3x_* + \frac{1+3v_r+v_r(1-v_r)}{1+v_r} T^r_{2k+2} + 2\left((1-v_r)-\frac{1+3v_r}{1+v_r}\right)T^r_{2k+2}\\
                                &= 3x_* + \frac{1+3v_r+v_r(1-v_r)+2(1+v_r)(1-v_r)-2(1+3v_r)}{1+v_r} T^r_{2k+2}\\
                                &= 3x_* + \frac{-(1+3v_r)+(2+3v_r)(1-v_r)}{1+v_r} T^r_{2k+2}\\
                                &= 3x_* + \frac{1-v_r(2+3v_r)}{1+v_r} T^r_{2k+2}\\
                        \end{align*}
                        and the competitive ratio is
                        \begin{align*}
                                \CR = 3 + \frac{1-v_r(2+3v_r)}{1+v_r} \frac{T^r_{2k+2}}{x_*}
                        \end{align*}
                        When $v_r(2+3v_r) \geq 1$ the competitive ratio is $\leq 3$. When $v_r(2+3v_r) < 1$ the competitive ratio is $> 3$ and increases with decreasing $x_*$, or, equivalently, with decreasing $t_*$. Thus, we should take $t_*$ arbitrarily close to $T^-_{2k+1}$. However, $t_* = T^-_{2k+1}$ gave the best-case evacuation time for the case that $t_* \leq T^-_{2k+1}$. We can thus conclude that a worst-case competitive ratio can be achieved when $t_* \leq T^-_{2k+1}$ and the competitive ratio of the algorithm is upper bounded by \eqref{eq:2s1r_cr}.
                \end{proof}
                
                Now that we have an expression for a bound on the competitive ratio we can finally prove Theorem~\ref{thm:2s_1r_ub}.
                \begin{proof}(Theorem~\ref{thm:2s_1r_ub}) 
                        We need to optimize the competitive ratio with respect to $v_r$ and so we need to compute the derivative of the right hand side of \eqref{eq:2s1r_cr}. This is most easily done with the aid of a computer. We find that
                        \begin{align*}
                                \frac{d}{dv_r}\left[1 + \frac{1+v_r}{1-v_r}\left(\frac{1+4v_r-v_r^2}{v_r(3-v_r)}\right)\right] = \frac{v_r^4 - 16v_r^3 + 26v_r^2 + 8v_r - 3}{v_r^2(3-v_r)^2(1-v_r)^2}
                        \end{align*}
                        and so the optimum choice of $v_r$ is a root of the quartic equation $v_r^4 - 16v_r^3 + 26v_r^2 + 8v_r - 3 = 0$ satisfying $0 \leq v_r < 1$. Numerically solving this equation for $v_r$ yields $v_r \approx 0.228652$. For this choice of $v_r$ one can confirm that $\CR < 5.681319$.
                \end{proof}


\section{Lower bounds}\label{sec:lower_bounds}
        In this section we investigate lower bounds on the competitive ratio of evacuation in our communication model. Our goal is the proof of the following theorem:
        \begin{theorem}\label{thm:lb}
                Let $\Alg$ be an evacuation algorithm for one sender and $n_r \geq 1$ receivers. 
                \[\CR(\Alg) \ge \begin{cases}
                        3+2\sqrt{2},& n_r=1\\
                        2+\sqrt{5},& n_r>1
                \end{cases}.\]
        \end{theorem}

        We will need to introduce a number of concepts and definitions. The first definition concerns the knowledge that is available to an agent at a given time.
        \begin{definition}
                An agent is said to {\em know of a location $x$ at time $t$} if it has direct or indirect knowledge of $x$ at time $t$. An agent with direct knowledge of $x$ at time $t$ has visited location $x$ at a time $t' \leq t$. An agent has indirect knowledge of $x$ at time $t$ if it can be notified of $x$ at a time $t' \leq t$. 
        \end{definition}
        The direct knowledge of an agent depends only on its own trajectory, whereas an agent's indirect knowledge depends on both its own and the other agents' trajectories. We define the direct knowledge set $K_X^D(t)$ as the set of all locations that the agent with trajectory $X$ has direct knowledge of at time $t$. We similarly define the indirect knowledge set $K^I_X(t;\Alg)$. The (total) knowledge set is the set $K_X(t;\Alg) = K^D_X(t) \cup K^I_X(t;\Alg)$. We make the following simple observation which results from the unit speed assumption.
        \begin{observation}
                $K^D_X(t) \subseteq [X(t)-t,X(t)+t]$.
        \end{observation}

        We can use the knowledge set of an agent to lower bound the competitive ratio.
        \begin{lemma}\label{lm:cr_general}
                For any evacuation algorithm $\Alg$, any $X \in \Alg$, and any time $t > 0$ we have
                \[\CR(\Alg) \geq \sup_{x \not\in K_X(t;\Alg)} \frac{|X(t)-x|+t}{|x|}.\]
        \end{lemma}
        Thus, we can derive a lower bound on the competitive ratio by bounding the size of an agent's knowledge set.

        We define the functional $\mu(X)$ which maps a search trajectory to a non-negative real number:
        \begin{equation}
                \mu(X) = :\limsup_{t \rightarrow \infty} \frac{|X(t)|}{t}.
        \end{equation}
        The quantity $\mu(X)$ can be thought of as an upper bound on the average rate at which the direct knowledge of an agent with trajectory $X$ grows. We naturally extend the definition of $\mu$ to take as input an evacuation algorithm:
        \begin{equation}
                \mu(\Alg) := \max_{X \in \Alg} \mu(X).
        \end{equation}


We establish several properties of $\mu(X)$ and $K_X$ (and relationships between them) for trajectories $X$ in evacuation algorithms. These are used in the proofs of the following theorems, from which Theorem~\ref{thm:lb} follows.

%
        \begin{lemma}\label{lm:KD}
                Let $X$ be a search trajectory. 
                If $\mu(X) < 1$ then for all $0 < \eps < 1-\mu(X)$ there exists a time $T > 0$ such that
                \[K^D_X(t) \subseteq \left[-\frac{(\mu(X)+\eps)(t-X(t))}{1+\mu(X)+\eps},\ \frac{(\mu(X)+\eps)(t+X(t))}{1+\mu(X)+\eps}\right],\ \forall t > T.\]
                In the case of $\mu(X) = 1$ the parameter $\epsilon$ can be taken to be $0$ in the above expression.
        \end{lemma}
        \begin{proof}
                
                By the definition of $\mu(X)$, it follows that for all $\eps > 0$ there exists a time $T'>0$ such that
                \begin{equation}\label{eq:1}
                        -(\mu(X)+\eps)t \leq X(t) \leq (\mu(X)+\eps)t,\ \forall t > T'.
                \end{equation}
                Moreover, when $\mu(X)=1$ the $\eps$ can be taken to be $0$, since $|X(t)| \le t$.
                
                In order to have direct knowledge of location $x$ at time $t > T'$ there must exist a time $t' \leq t$ such that $X(t') = x$. The unit speed of the agents implies that $|X(t)-x| \leq t-t'$ or
                \begin{equation}\label{eq:2}
                        t'-t+X(t) \leq x \leq t+X(t)-t'.
                \end{equation}
                Assume that $t' > T'$. Then we can combine \eqref{eq:1} and \eqref{eq:2} to get
                \begin{equation}\label{eq:3}
                        \max\{-(\mu(X)+\eps)t',\ t'-t+X(t)\} \leq x \leq \min\{(\mu(X)+\eps)t',\ t+X(t)-t'\}.
                \end{equation}
                On the left, the first term in the max decreases with $t'$ and the second term increases with $t'$. Thus, the best lower bound is achieved when the two terms are equal. This will occur when
                \[t' = \frac{t-X(t)}{1+\mu(X)+\eps}.\]
                For this value of $t'$ we get
                \[x \geq -\frac{(\mu(X)+\eps)(t-X(t))}{1+\mu(X)+\eps}.\]
                At time $t$ we have $X(t) \leq (\mu(X)+\eps)t$ and thus
                \[t' = \frac{t-X(t)}{1+\mu(X)+\eps} \geq  \frac{t-(\mu(X)+\eps)t}{1+\mu(X)+\eps} = \frac{1-\mu(X)-\eps}{1+\mu(X)+\eps}t.\]
                Hence, we will have $t' > T'$ for all
                \begin{equation}\label{eq:4}
                        t > T = \frac{1+\mu(X)+\eps}{1-\mu(X)-\eps}T'.
                \end{equation}
                When $\mu(X) = 1$ the above expression is vacuously true, since $T'$ can be taken to be $0$.

                In a similar manner, we get from the right side of \eqref{eq:3} that
                \[x \leq \frac{(\mu(X)+\eps)(t+X(t))}{1+\mu(X)+\eps}.\]
                for all $t$ satisfying \eqref{eq:4}. This completes the proof.
        \end{proof}

        In a similar manner we can bound the total knowledge available to an agent that can only receive messages face-to-face.

        \begin{lemma}\label{lm:Kf2f}
                Let $\Alg$ be an evacuation algorithm and let $X_{f2f} \in \Alg$ represent the trajectory of an agent that can only receive messages face-to-face. 
                If $\mu(X_{f2f})<1$ then for all $0 < \eps < 1-\mu(X_{f2f})$ there exists a time $T > 0$ such that
                \[K_{X_{f2f}}(t;\Alg) \subseteq \left[-\frac{(\mu(\Alg)+\eps)(t-X_{f2f}(t))}{1+\mu(\Alg)+\eps},\ \frac{(\mu(\Alg)+\eps)(t+X_{f2f}(t))}{1+\mu(\Alg)+\eps}\right],\ \forall t > T.\]
                In the case of $\mu(X_{f2f}) = 1$ the parameter $\epsilon$ can be taken to be $0$ in the above expression.
        \end{lemma}
        \begin{proof}

                When $\mu(X) < 1$ it follows from the definitions of $\Alg$ and $\mu(X)$ that there exists a time $T'>0$ such that for any $X \in \Alg$ we have
                \begin{equation}\label{eq:5}
                        -(\mu(\Alg)+\eps)t \leq X(t) \leq (\mu(\Alg)+\eps)t,\ \forall t > T',\ \forall X \in \Alg.
                \end{equation}
                In order to have direct knowledge of location $x$ at time $t > T'$ there must exist a time $t' \leq t$ such that $X_{f2f}(t') = x$. The unit speed of the agents implies that $|X_{f2f}(t)-x| \leq t-t'$ or that \eqref{eq:2} must be satisfied by the trajectory $X_{f2f}$ at time $t$.

                In order to have indirect knowledge of location $x$ at time $t$ there must exist another agent that visits $x$ at time $t' \leq t$ and can reach location $X_{f2f}(t)$ by time $t$. Indeed, this agent must be able to catch the agent with trajectory $X_{f2f}$ at or before time $t$, and the agent with trajectory $X_{f2f}$ will be at location $X_{f2f}(t)$ at time $t$. Thus, in order to have indirect knowledge of $x$, the unit speed condition implies again that $|X_{f2f}(t)-x| \leq t-t'$ or that \eqref{eq:2} is satisfied. To complete the proof we follow the same steps of the proof of Lemma~\ref{lm:KD} except with \eqref{eq:5} used in place of \eqref{eq:1}.
        \end{proof}     
        
        We will now focus on the case that there is only a single sender involved in the evacuation. The sender can only be communicated with face-to-face and so Lemma~\ref{lm:Kf2f} applies in this case. We can use it to get the following result.
        \begin{lemma}\label{lm:1s_1}
                Let $\Alg$ be an evacuation algorithm with one sender and let $S \in \Alg$ represent the trajectory of this sender. If $\mu(S)=1$ then $\CR(\Alg)$ is unbounded. If $\mu(S)<1$ then we have 
                \[\CR(\Alg) \ge 1 + \frac{(1+\mu(\Alg))(1+\mu(S))}{\mu(\Alg)(1-\mu(S))}.\]
        \end{lemma}
        \begin{proof}
                If $\mu(S)=1$ then the previous lemma tells us that $K_{S}(t;\Alg) = [-\frac{t-S(t)}{2},\frac{t+S(t)}{2}]$. Suppose without loss of generality that $S(t) = t$. Then $K_S(t) = [0,t]$ and by Lemma~\ref{lm:cr_general} we have 
                \[\CR(\Alg) \geq \sup_{x \not\in K_S(t;\Alg)} \frac{|X(t)-x|+t}{|x|} = \sup_{\eps > 1} \frac{2t+\eps}{\eps} > 2t+1.\]
                Since the above holds for infinitely many arbitrary large $t$ values, we conclude that $\CR(\Alg)$ is unbounded.
                
                Suppose now that $\mu(S)<1$. Then for all $0 < \eps < 1-\mu(S)$ there exists a time $T$ such that
                \[K_{S}(t;\Alg) \subseteq \left[-\frac{(\mu(\Alg)+\eps)(t-S(t))}{1+\mu(\Alg)+\eps},\ \frac{(\mu(\Alg)+\eps)(t+S(t))}{1+\mu(\Alg)+\eps}\right],\ \forall t > T.\]
                Moreover, from the definition of $\mu(X)$ it follows that for any $\Delta > 0$ there exists a time $\tau$ such that $|S(\tau)| = \mu(S)\tau \pm o(\tau)$. Take $\Delta > T$ and assume without loss of generality that $S(\tau) = \mu(S) \tau \pm o(\tau)$. Then
                \[K_{S}(\tau;\Alg) \subseteq \left[-\frac{(\mu(\Alg)+\eps)(1-\mu(S))\tau}{1+\mu(\Alg)+\eps},\ \frac{(\mu(\Alg)+\eps)(1+\mu(S))\tau}{1+\mu(\Alg)+\eps}\right]\]
                and
                \begin{align*}
                        \CR(\Alg) &\geq \sup_{x \not\in K_S(\tau;\Alg)} \frac{|S(t)-x|+t}{|x|} \\ &= 1 + \sup_{\eps' > 0} \frac{(1+\mu(S))\tau}{\frac{(\mu(\Alg)+\eps)(1-\mu(S))\tau}{1+\mu(\Alg)+\eps}+\eps'} =  1 + \frac{(1+\mu(\Alg)+\eps)(1+\mu(S))}{(\mu(\Alg)+\eps)(1-\mu(S))}\\
                        &> 1 + \frac{(1+\mu(\Alg))(1+\mu(S))}{\mu(\Alg)(1-\mu(S))}\left[\frac{1}{1+\frac{\eps}{\mu(\Alg)}}\right].
                \end{align*}
                The term in square brackets approaches 1 from below as $\eps \rightarrow 0$ and thus for any fixed $\delta>0$ we can choose $\eps>0$ small enough that
                \begin{align*}
                        \CR(\Alg) &> 1 - \delta + \frac{(1+\mu(\Alg))(1+\mu(S))}{\mu(\Alg)(1-\mu(S))}.
                \end{align*}
        \end{proof}
        \begin{corollary}
                Let $\Alg$ be an evacuation algorithm with one sender and let $S \in \Alg$ represent the trajectory of this sender. If $\mu(S) = \mu(\Alg)$ we have $\CR(\Alg) \ge 9$.
        \end{corollary}
        \begin{proof}
                With $\mu(S) = \mu(\Alg)$ we have from Lemma~\ref{lm:1s_1} that
                \[\CR(\Alg) \ge 1 + \frac{(1+\mu(S))^2}{\mu(S)(1-\mu(S))}\]
                for all $\delta > 0$. Let $g(u) = \frac{(1+u)^2}{u (1-u)}$ and observe that
                \begin{align*}
                        \frac{dg(u)}{du} &= \frac{2(1+u)}{u(1-u)} - \frac{(1+u)^2(1-2u)}{u^2(1-u)^2} = (1+u)\left[\frac{2u(1-u) - (1+u)(1-2u)}{u^2(1-u)^2}\right]\\
                        &= (1+u)\left[\frac{2u-2u^2 - 1 + 2u - u + 2u^2}{u^2(1-u)^2}\right]
                        = \frac{(1+u)(3u-1)}{u^2(1-u)^2}.
                \end{align*}
                From this last expression it is clear that $u = 1/3$ is the only non-negative minimizer. When $u=1/3$ we find $g(1/3) = \frac{(1+\frac{1}{3})^2}{\frac{1}{3}(1-\frac{1}{3})} = 8$ and thus we can conclude that $\CR(\Alg) > 9 - \delta$ for arbitrary $\delta > 0$as required.
        \end{proof}

        In the next lemma we consider the knowledge set of the receivers.
        \begin{lemma}\label{lm:Kreceivers}
                Let $\Alg$ be an evacuation algorithm with one sender and at least one receiver. Let $S$ be the trajectory of the sender and suppose that $\mu(S) < \mu(\Alg)$. Let $R$ be the trajectory of a receiver with $\mu(R) = \mu(\Alg)$. 
                If $\mu(R) < 1$ then for all $0 < \eps < 1-\mu(R)$ there exists a time $T>0$ such that
                \[K_{R}(t;\Alg) = K_S(t;\Alg) \cup \left[-\frac{(\mu(R)+\eps)(t-R(t))}{1+\mu(R)+\eps},\ \frac{(\mu(R)+\eps)(t+R(t))}{1+\mu(R)+\eps}\right],\ \forall t > T.\]
                In the case of $\mu(R) = 1$ the parameter $\epsilon$ can be taken to be $0$ in the above expression.
        \end{lemma}
        \begin{proof}
                The receivers can receive wireless messages from the sender and so at any time they know what the sender knows. If we exclude knowledge from the sender, a receiver can only possess direct knowledge, or receive knowledge indirectly from a different receiver. However, receivers can't send messages and so communication between receivers is face-to-face. Thus, to complete the proof, we only need to invoke Lemma~\ref{lm:Kf2f}.
        \end{proof}

        \begin{lemma}\label{lm:cr_r}
                Let $\Alg$ be an evacuation algorithm with one sender and at least one receiver. Let $S \in \Alg$ represent the trajectory of this sender; let $R \in \Alg$ represent the trajectory of the receiver with the largest value of $\mu(R)$; and define $\Alg' = \Alg \setminus \{R\}$. Then, we have
                \begin{align*}
                        \CR(\Alg) &\ge 1 + \frac{(1+\mu(\Alg'))(1+\mu(R))}{\mu(\Alg')(1+\mu(S))}.
                \end{align*}
        \end{lemma}
        \begin{proof}
                We make use of Lemma~\ref{lm:Kreceivers}. If $\mu(R)=1$ then
                \[K_{R}(t;\Alg) = K_S(t;\Alg) \cup \left[-\frac{t-R(t)}{2},\frac{t+R(t)}{2}\right].\]
                If $\mu(R) < 1$ then for all $0 < \eps < 1-\mu(R)$ there exists a time $T_R>0$ such that
                \[K_{R}(t;\Alg) = K_S(t;\Alg) \cup \left[-\frac{(\mu(R)+\eps)(t-R(t))}{1+\mu(R)+\eps},\ \frac{(\mu(R)+\eps)(t+R(t))}{1+\mu(R)+\eps}\right],\ \forall t > T_R.\]
                Moreover, for any $\Delta > 0$ there exists a time $\tau > \Delta$ such that $|R(t)| = \mu(R)\tau$. Assume without loss of generality that $R(\tau) = \mu(R)\tau$. Then
                $K_{R}(\tau;\Alg) = K_S(\tau;\Alg) \cup \left[0,\tau\right]$ if $\mu(R)=1$ and for $\mu(R)<1$ we can take any $\Delta > T_R$ to get
                \[K_{R}(\tau;\Alg) = K_S(\tau;\Alg) \cup \left[-\frac{(\mu(R)+\eps)(1-\mu(R))\tau}{1+\mu(R)+\eps},\ \frac{(\mu(R)+\eps)(1+\mu(R))\tau}{1+\mu(R)+\eps}\right].\]
                In light of the proof of Lemma~\ref{lm:1s_1} and its corollary, it is clear that unless the sender can increase the lower bound of the receiver's knowledge, we will find that $\CR(\Alg)$ is unbounded when $\mu(R)=1$, and when $\mu(R)<1$ we will get $\CR(\Alg) > 9-\delta$ for all $\delta > 0$. Thus, we assume that the sender can increase the lower bound of the receiver's knowledge.

                Consider the knowledge set $K_S(t;\Alg)$ of the sender. Since this must extend the knowledge of the receiver with trajectory $R$ we can exclude this receiver from the computation of $K_S(t;\Alg)$. Thus, if we take $\Alg' = \Alg \setminus \{R\}$ we can invoke Lemma~\ref{lm:Kf2f} with respect to $\Alg'$ to conclude that for all $0 < \eps < 1-\mu(S)$ there exists a time $T_S > 0$ such that
                \[K_{S}(t;\Alg') \subseteq \left[-\frac{(\mu(\Alg')+\eps)(t-S(t))}{1+\mu(\Alg')+\eps},\ \frac{(\mu(\Alg')+\eps)(t+S(t))}{1+\mu(\Alg')+\eps}\right],\ \forall t > T_S.\]
                We can take $\Delta > \max\{T_R,T_S\}$ so that $\tau > T_S$ and as a result
                \[K_{S}(\tau;\Alg) \subseteq \left[-\frac{(\mu(\Alg')+\eps)(\tau-S(\tau))}{1+\mu(\Alg')+\eps},\ \frac{(\mu(\Alg')+\eps)(\tau+S(\tau))}{1+\mu(\Alg')+\eps}\right].\]
                By definition of $\mu(S)$, at any time $t > T_s$ we have $|S(t)| \leq (\mu(S)+\eps)t$ and thus
                \[K_{S}(\tau;\Alg) \subseteq \left[-\frac{(\mu(\Alg')+\eps)(1+\mu(S)+\eps)\tau}{1+\mu(\Alg')+\eps},\ \frac{(\mu(\Alg')+\eps)(1+\mu(S)+\eps)\tau}{1+\mu(\Alg')+\eps}\right].\]
                By Lemma~\ref{lm:cr_general} we then have
                \begin{align*}
                        \CR(\Alg) &\geq \sup_{x \notin K_{R}(\tau;\Alg)} \frac{|R(\tau)-x|+\tau}{|x|} \\ &= 1 + \sup_{\eps'>0}\frac{(1+\mu(R))\tau}{\frac{(\mu(\Alg')+\eps)(1+\mu(S))\tau}{1+\mu(\Alg')+\eps}+\eps'} = 1 + \frac{(1+\mu(\Alg')+\eps)(1+\mu(R))}{(\mu(\Alg')+\eps)(1+\mu(S)+\eps)}\\
                        &> 1 + \frac{(1+\mu(\Alg'))(1+\mu(R))}{\mu(\Alg')(1+\mu(S))}\left[\frac{1}{1+\frac{\eps}{\mu(\Alg')}}\right]\left[\frac{1}{1+\frac{\eps}{\mu(S)}}\right].
                \end{align*}
                It is clear that $\mu(S) > 0$ (and, thus, also $\mu(\Alg')>0$) since otherwise the sender would not extend the receiver's knowledge. Then both of the terms in square brackets approach 1 from below as $\eps \rightarrow 0$ and so we can choose $\eps>0$ small enough that for any fixed $\delta > 0$ we have
                \begin{align*}
                        \CR(\Alg) &> 1 - \delta + \frac{(1+\mu(\Alg'))(1+\mu(R))}{\mu(\Alg')(1+\mu(S))}
                \end{align*}
                as required.
        \end{proof}

        \begin{theorem}\label{thm:cr_11}
                Let $\Alg$ be an evacuation for one sender and one receiver. Then $\CR(\Alg) \ge 3+2\sqrt{2}$.
        \end{theorem}
        \begin{proof}
                Let $\Alg = \{S,R\}$ with $S$ and $R$ the trajectories of the sender and receiver respectively. We must have $\mu(S)<\mu(\Alg)=\mu(R)$ since otherwise the competitive ratio is at least $9$. Then, Lemma~\ref{lm:1s_1} states that
                \[\CR(\Alg) \ge 1  + \frac{(1+\mu(R))(1+\mu(S))}{\mu(R)(1-\mu(S))}\]
                and from Lemma~\ref{lm:cr_r} we have
                \begin{align*}
                        \CR(\Alg) &\ge 1 + \frac{(1+\mu(S))(1+\mu(R))}{\mu(S)(1+\mu(S))} = \frac{1+\mu(R)}{\mu(S)}
                \end{align*}
                where we have used the fact that $\Alg' = \{S\}$ when there is only one receiver. We therefore have
                \[\CR(\Alg) \ge 1 + \max\left\{\frac{(1+\mu(R))(1+\mu(S))}{\mu(R)(1-\mu(S))},\ \frac{1+\mu(R)}{\mu(S)}\right\}.\]
                The first term in the max decreases with $\mu(R)$ and the second term increases with $\mu(R)$ and so our best-lower bound is achieved when increasing  $\mu(R)$ is such that the two terms are equal. We find that we need
                \[\mu(R) = \frac{\mu(S)(1+\mu(S))}{1-\mu(S)}.\]
                We then find that
                \[\CR(\Alg) \ge 1 +\frac{1+\frac{\mu(S)(1+\mu(S))}{1-\mu(S)}}{\mu(S)} = 1 + \frac{1}{\mu(S)}+ \frac{1+\mu(S)}{1-\mu(S)}.\]
                Let $g(u) = \frac{1}{u}+ \frac{1+u}{1-u}$ and observe that
                \begin{align*}
                        \frac{dg(u)}{du} &= -\frac{1}{u^2} + \frac{1}{1-u} + \frac{1+u}{(1-u)^2} \\ &= -\frac{1}{u^2} + \frac{2}{(1-u)^2} = \frac{(1-u)^2-2u^2}{u^2(1-u)^2}
                        = \frac{(1-u-\sqrt{2}u)(1-u+\sqrt{2}u)}{u^2(1-u)^2}\\
                        &= \frac{[1-(1+\sqrt{2})u][1-(1-\sqrt{2})u]}{u^2(1-u)^2}.
                \end{align*}
                From this last expression it is clear that $g(u)$ is minimized when $u=\frac{1}{1+\sqrt{2}} = \sqrt{2}-1$. The minimum is
                \[g(\sqrt{2}-1) = \sqrt{2}+1 + \frac{\sqrt{2}}{2-\sqrt{2}} = 2+2\sqrt{2}\]
                and we can conclude that
                \[\CR(\Alg) \ge 3+2\sqrt{2}.\]
        \end{proof}

\begin{corollary} The evacuation algorithm for one sender and one receiver given in the proof of Theorem \ref{thm:1s_1r_ub} is optimal.
\end{corollary}

        \begin{theorem}\label{thm:cr_12}
                Let $\Alg$ be an evacuation for one sender and $n_r>1$ receivers. Then $\CR(\Alg) \ge 2+\sqrt{5}$.
        \end{theorem}
        \begin{proof}
                Let $\Alg = \{S,R,R',\ldots\}$ with $S$ the trajectory of the sender, $R$ the trajectory of the receiver with largest value of $\mu(R)$, and $R'$ the trajectory of the receiver with second largest $\mu(R')$. We must have $\mu(S)<\mu(\Alg)=\mu(R)$ since otherwise the competitive ratio is at least $9$. Then, Lemma~\ref{lm:1s_1} states that
                \[\CR(\Alg) \ge 1 + \frac{(1+\mu(R))(1+\mu(S))}{\mu(R)(1-\mu(S))}\]
                and from Lemma~\ref{lm:cr_r} we have
                \begin{align*}
                        \CR(\Alg) &\ge 1  + \frac{(1+\mu(\Alg'))(1+\mu(R))}{\mu(\Alg')(1+\mu(S))}.
                \end{align*}
                If $\mu(\Alg') = \mu(S)$ then it follows from Theorem~\ref{thm:cr_11} that we will have $\CR(\Alg) \ge 3+2\sqrt{2}$. Thus, we assume that $\mu(S) < \mu(\Alg')$. Then $\mu(R') = \mu(\Alg')$ and we have
                \[\CR(\Alg) \ge 1 + \max\left\{\frac{(1+\mu(R))(1+\mu(S))}{\mu(R)(1-\mu(S))},\ \frac{(1+\mu(R'))(1+\mu(R))}{\mu(R')(1+\mu(S))}\right\}.\]
                The second term in the max increases with decreasing $\mu(R')$ and the first term does not depend on $\mu(R')$. Thus, we set $\mu(R')$ as large as possible, i.e., we take $\mu(R') = \mu(R)$. Then
                \[\CR(\Alg) \ge 1 + \max\left\{\frac{(1+\mu(R))(1+\mu(S))}{\mu(R)(1-\mu(S))},\ \frac{(1+\mu(R))^2}{\mu(R)(1+\mu(S))}\right\}.\]
                Now both terms in the max increase with decreasing $\mu(R)$ and so we take $\mu(R)$ as large as possible, i.e., $\mu(R)=1$. Then
                \[\CR(\Alg) \ge 1 + 2\max\left\{\frac{1+\mu(S)}{1-\mu(S)},\ \frac{2}{1+\mu(S)}\right\}.\]
                The first term in the max increases with $\mu(S)$ and the second term decreases with $\mu(S)$ and so our best-lower bound is achieved when $\mu(S)$ is such that the two terms are equal. We find that we need $(1+\mu(S))^2 = 2(1-\mu(S))$ or
                \[\mu(S)^2+4\mu(S)-1=0.\]
                The only non-negative solution to this quadratic equation is
                \[\mu(S) = \frac{-4 + \sqrt{20}}{2} = \sqrt{5}-2\]
                and we can conclude that
                \[\CR(\Alg) \ge 1 + \frac{4}{\sqrt{5}-1} = 1  + \frac{4(\sqrt{5}+1)}{4} = 2+\sqrt{5}\]
                as required.
        \end{proof}        

        Our upper bound for the case that $n_s=1$ and $n_r>1$ was $5 > 2+\sqrt{5}$ and so either the lower bound is not tight and/or the upper bound must come down. If one refers to the proof of Theorem~\ref{thm:cr_12} then one can observe that our best lower bound was achieved when $\mu_s = \sqrt{5}-2$ and $\mu_1=\mu_2=1$. However, one can easily confirm that any algorithm with $\mu_1=\mu_2=1$ has a competitive ratio of at least 5 and so it is evident that, at least, the lower bound is not tight. Thus, in order to make progress on this problem, we believe a different lower bounding technique will be required.

\section{Conclusions}\label{sec:conclusions}
        We have introduced a novel communication model that puts an interesting twist on the classic linear group search problem. We provide upper bounds on the evacuation for the three interesting combinations of agents -- one sender and one receiver, one sender and multiple receivers, and multiple senders and one receiver. We demonstrate that our algorithm for the case of one sender and one receiver is optimal by providing a lower bound matching our upper bound. For the case of one sender and two receivers we provide a non-trivial lower bound of $2+\sqrt{5}$ which compares to our upper bound of 5. We do not provide any non-trivial lower bounds for the case of multiple senders and one receiver and it is believed that this is the most difficult case to do so (indeed, the upper bound for this case was considerably more complex than the other two cases).

        The most immediate open problems concern the lower bounds for the cases of multiple senders and multiple receivers. For the multiple receiver case we provided arguments demonstrating that the lower bound presented here cannot be tight and so in order to close the gap between the lower and upper bounds a different lower bounding technique will be required. Of course, it can also be the case that the upper bound must come down as well (although this does not seem likely). We did not attempt to provide a lower bound for the case of multiple senders (there is a trivial lower bound of 3 which can be derived by considering the first time any agent reaches location $\pm x$).

        Our upper bounds on the evacuation seem to hint at the fact that it is better to ``listen'' than it is to ``speak'' since our upper bound for the case of multiple receivers is 5 and for multiple senders it is $\approx 5.681319$ (and we do not believe that these can be improved). Closing the gap between the upper and lower bounds would be interesting even just from the standpoint of answering the question of whether or not it is better to ``listen'' than it is to ``speak''.


\bibliographystyle{plainurl}
\bibliography{refs}

\end{document}